\documentclass[a4paper,twoside,10pt]{amsart}
\usepackage{amsmath}
\usepackage{amsfonts}
\usepackage{amssymb}
\usepackage{amsthm}
\usepackage{newlfont}
\usepackage{graphicx}
\usepackage{amscd}

\theoremstyle{plain}
\newtheorem{Th}{Theorem}%[section]
\newtheorem{Cor}[Th]{Corollary}

\newtheorem{Prop}[Th]{Proposition}
\theoremstyle{definition}
\newtheorem{Def}{Definition}%[section]
\theoremstyle{remark}
\newtheorem*{Rem}{Remark}%[section]
\newcommand{\PP}{{\mathbb P}}

\newcommand{\HH}{{\mathbb H}}
\newcommand{\CC}{{\mathbb C}}

\newcommand{\KK}{{\mathbb K}}
\newcommand{\FF}{{\mathbb F}}

\newcommand{\ZZ}{{\mathbb Z}}
\newcommand{\VV}{{\mathbb V}}
\newcommand{\RR}{{\mathbb R}}

\newcommand{\WW}{{\mathbb W}}

\newcommand{\bPhi}{\boldsymbol{\Phi}}

\begin{document}

\title[Quadrangular sets and discrete Schwarzian KP equation]
{Quadrangular sets in projective line  and in Moebius space,\\ and geometric interpretation of the non-commutative discrete Schwarzian Kadomtsev--Petviashvili  equation}

\author[A. Doliwa]{Adam Doliwa}

\address{A. Doliwa, Faculty of Mathematics and Computer Science, University of Warmia and Mazury,
ul.~S{\l}oneczna~54, 10-710~Olsztyn, Poland}

\email{doliwa@matman.uwm.edu.pl}
\urladdr{http://wmii.uwm.edu.pl/~doliwa/}

\author[J. Kosiorek]{Jaros{\l}aw Kosiorek}

\address{J. Kosiorek, Faculty of Mathematics and Computer Science, University of Warmia and Mazury,
ul.~S{\l}oneczna~54, 10-710~Olsztyn, Poland}

\email{kosiorek@matman.uwm.edu.pl}

\date{}

\begin{abstract}
We present geometric interpretation of the discrete Schwarzian Kadomtsev--Petviashvili equation in terms of quadrangular set of points of a projective line. We give also the corresponding interpretation for the projective line considered as a Moebius chain space. In this way we incorporate the conformal geometry interpretation of the equation into the projective geometry approach via Desargues maps.
\end{abstract}
\maketitle

\section{Introduction}
In the present paper we address two questions concerning geometric interpretation of the following discrete integrable system
\begin{equation} \label{eq:NCSKP}
(\phi_{(jk)} - \phi_{(k)})(\phi_{(jk)} - \phi_{(j)})^{-1}
(\phi_{(ij)} - \phi_{(j)})(\phi_{(ij)} - \phi_{(i)})^{-1}
(\phi_{(ik)} - \phi_{(i)})(\phi_{(ik)} - \phi_{(k)})^{-1} = 1, 
\end{equation}
where $\phi\colon\ZZ^N\to\FF$ is a map from $N$-dimensional integer lattice to a division ring $\FF$, and indices in brackets denote shifts in the corresponding variables, i.e.
$
\phi_{(i)}(n_1,\dots , n_N) = \phi(n_1, \dots , n_i + 1, \dots n_N).
$
The above equation appeared first as the generalized lattice spin equation in~\cite{FWN-Capel}, and was called the non-commutative discrete Schwarzian Kadomtsev--Petviashvili (SKP) equation in \cite{BoKo-N-KP,KoSchiefSDS-II}. 
As being one of various forms of the discrete Kadomtsev-Petviashvili (KP) system~\cite{Hirota,KNS-rev}, equation~\eqref{eq:NCSKP} plays pivotal role in the theory of integrable systems and its applications.

Relation between geometry of submanifolds and integrable systems is an ongoing research subject which can be dated back to second half of XIX-th century~\cite{DarbouxIV}. In fact, geometric approach to discrete integrable systems initiated in~\cite{BP2,MQL,KoSchief2}, see also \cite{BobSur} for a review, demonstrates that the basic principles of the theory are encoded in incidence geometry statements, some of them known in antiquity.

For example, complex $\FF = \CC$ version of equation~\eqref{eq:NCSKP}  was identified in~\cite{KoSchief-Men} as a multi-ratio condition which describes generalization to conformal geometry of circles of the Menelaus theorem in the metric geometry~\cite{Coxeter-Greitzer}. Quaternionic version of the equation was studied in~\cite{KoSchiefSDS-II,KoSchief-qSKP}, see also \cite{King-Schief-mult} for other geometric interpretations of the multi-ratio condition in relation to integrable discrete systems. 

The more recently introduced notion of Desargues maps \cite{Dol-Des}, as underlying property of discrete KP equation considers collinearity of three points. This approach works in projective geometries over division rings and leads directly to the linear problem for the equation in its non-Abelian Hirota--Miwa form~\cite{Nimmo-NCKP}. We remark that that the Desargues maps give new understanding ~\cite{Dol-Des,Dol-Des-red} of the previously studied discrete conjugate nets \cite{MQL}. These are characterized by planarity of  elementary quadrilaterals (see also \cite{Sauer2,Doliwa-T}). The compatibility condition for Desargues maps gives projective Menelaus theorem, but leaves open the following \textbf{Question 1: Can the conformal geometry interpretation of the discrete Schwarzian Kadomtsev--Petviashvili equation be incorporated into the Desargues map approach?} Notice that the recent generalization of the Desargues theorem to  context of conformal geometry~\cite{KingSchief-Des} may suggest something opposite. 

When studying reductions of the Desargues maps, as for example in~\cite{Doliwa-GD}, one is forced to restrict dimension of the ambient projective space up to ``Desargues maps into projective line''. Even if the linear problem is well defined there the geometric condition, which defines the maps, is empty. This leads to \textbf{Question 2: What should replace the Desargues map condition for the ambient space being projective line?} We remark that the analogous problem for discrete conjugate nets in the ambient space being a plane was successfully solved in~\cite{Adler}.

Our answer for both questions is based on the notion of the quadrangular sets, which was introduced by von~Staudt in his seminal work~\cite{Staudt} as a tool to provide axiomatization of the projective geometry. We remark that quadrangular sets of points appeared in integrable discrete geometry in theory of the $B$-quadrilateral lattice \cite{Doliwa-BQL}, but in the context of the Pappus theorem and the Moebius pair of tetrahedra, which is outside of the interest of the present paper. 

In Section \ref{sec:line} we first recall basic ingredients of the geometry of the projective line, in particular the notion of cross-ratio in the general non-commutative case~\cite{Baer}. We also formulate the corresponding concept of the multi-ratio of six points on the projective line over a division ring, which generalizes the definition known for commutative case in terms of two cross-ratios or determinants. We show that quadrangular sets of points are fully characterized by the ``multiratio equals one'' condition also in the non-commutative case (as the commutative case is well known \cite{Richter-Gebert}). This gives our answer to Question~2, which we present in Section~\ref{sec:Des-line}.

Our answer to Question~1, which we present in Section~\ref{sec:CG}, is also implied by geometry of the projective line, but this time the line is equipped with additional structure.  When the division ring $\FF$ contains a subfield $\KK$ in its center then $\FF$-projective images of the canonically embedded $\KK$-line form the so called chains. This leads to the concept of Moebius chain geometry \cite{Benz,Herzer}. We show that in such spaces certain quadrangular sets have particular interpretation in terms of the so called Moebius--Veblen chain configuration. 
In the simplest case of the classical Moebius geometry, where $\KK=\RR\subset\CC=\FF$ the chains are circles (homographic images of the real line in the complex conformal plane), and our approach gives that of Konopelchenko and Schief~\cite{KoSchief-Men}.

\section{Projective geometry of a line} \label{sec:line}
\subsection{Cross-ratio and multi-ratio in projective geometry over division rings}
A \emph{right linear space} $(\mathbb{F},\VV)$ consists of a division ring $\mathbb{F}$  and an additive abelian group $\VV$ such that $\mathbb{F}$ acts on $\VV$ from the right satisfying usual axioms. The corresponding projective geometry studies linear subspaces of the $\mathbb{F}$-space $\VV$. The points of the corresponding projective space $\PP(\FF,\VV)$ are one dimensional subspaces of $(\FF,\VV)$.
\begin{Rem}
For simplicity we assume that $\FF$ is of characteristic zero,  but we expect that also finite characteristic may be relevant and give interesting results~\cite{BiaDol}.
\end{Rem}

\emph{A collineation of the linear space  $(\mathbb{F},\VV)$ upon the linear  space  $(\mathbb{F},\WW)$} is a bijective and order preserving mapping $\sigma$ of the partially ordered (by inclusion) set of subspaces of $\VV$ upon the set of subspaces of $\WW$. 
When dimension of $\VV$ is at least three, any such collineation is given by a semi-linear map, i.e. linear map $\VV \to \WW$ and supplemented by an automorphism of the division ring. 

The case of two dimensional linear spaces (i.e. projective lines) needs a special treatment. Then any bijection of projective line can be called collineation. There arises the problem of characterizing those maps which are induced by semi-linear maps of $(\mathbb{F},\VV)$. The first step in that direction (the full answer can be found in~\cite{Baer}) makes use of a generalization of the classical notion of cross ratio. 
\begin{Def}[\cite{Baer}]
	Suppose that $P,Q,R,S$ are four distinct points on the line $L$. Then the number $c\in\mathbb{F}$ belongs to the cross ratio $\left[ \begin{array}{cc} P & Q \\ S & R \end{array} \right] \subset \FF$ if there exist elements $\boldsymbol{p} ,\boldsymbol{q} \in\VV$ such that
	\begin{equation*}
	P = \langle \boldsymbol{p} \rangle , \qquad Q = \langle \boldsymbol{q} \rangle , \qquad R = \langle \boldsymbol{p} + \boldsymbol{q}  \rangle , \qquad S = \langle \boldsymbol{p} + \boldsymbol{q} c \rangle .  
	\end{equation*}
\end{Def}
Below we present the  known expression of the cross-ratio in terms of non-homogeneous coordinates.
\begin{Th}
If $\boldsymbol{k} ,\boldsymbol{l}$ are two independent elements of $ \in\VV$, and 
$p,q,r,s$ are four distinct elements of $\FF$ then
\begin{equation*}
\left[ \begin{array}{cc} 
\langle \boldsymbol{k} + \boldsymbol{l}p \rangle & 
\langle \boldsymbol{k} + \boldsymbol{l}q \rangle \\ 
\langle \boldsymbol{k} + \boldsymbol{l}s \rangle & 
\langle \boldsymbol{k} + \boldsymbol{l}r \rangle \end{array} \right] =
\left[ (s-q)^{-1} (p-s) (p-r)^{-1} (r-q) \right],
\end{equation*}	
where by for $c\in \FF$ by $[c] = \{  aca^{-1} \, | \,  a\in\FF\setminus\{0\}\} $
we denote the equivalence class of conjugate elements.
\end{Th}
\begin{Rem}
Given three distinct points $P,Q,R\in L$ and given $c\in \FF\setminus\{0,1\}$ there exists the fourth point $S\in L$, distinct from the previous ones, such that $c \in \left[ \begin{array}{cc} P & Q \\ S & R \end{array} \right]$;
$c=0$ corresponds to $S=P$, $c=1$ corresponds to $S=R$, while to admit $S=Q$ we need to give the value $c=\infty$.
\end{Rem}
\begin{Rem}
	When points $P,Q,R$ have been taken as projective basis of the line, i.e. $p=0,q=\infty, r = 1$, then $\left[ \begin{array}{cc} P & Q \\ S & R \end{array} \right]= [s]$.
\end{Rem}
The following result for $\dim \VV \geq 3$ justifies the use of cross-ratio in describing geometry of the projective line.
\begin{Th} Suppose that $(P,Q,R,S)$ and $(P^\prime,Q^\prime,R^\prime, S^\prime)$
	are quadruples of distinct collinear points.
	There exists collineation $\pi$ of the linear space $(\FF,\VV)$, $\dim \VV \geq 3$. such that $\pi(P) = P^\prime$, $\pi(Q) = Q^\prime$, $\pi(R) = R^\prime$, $\pi(S) = S^\prime$ if and only if there exists an automorphism $\alpha$ of the division ring $\FF$ such that 
	\begin{equation*}
	\left[ \begin{array}{cc} P & Q \\ S & R \end{array} \right]^\alpha = \left[ \begin{array}{cc} P^\prime & Q^\prime \\ S^\prime & R^\prime \end{array} \right].
	\end{equation*}
\end{Th}

We present below the analogous geometric notion of multi-ratio in the non-commutative case, which we adapted from known definition in the commutative case in terms of two cross-ratios~\cite{MorleyMusselman,King-Schief-mult,Richter-Gebert}. Like for the non-commutative cross-ratio our geometric definition leads to a class of conjugate elements of the division ring.

\begin{Def} \label{def:multi-ratio}
	Suppose that $P,Q,R,S,T,U$ are six distinct points on the line $L$. Then the number $m\in\mathbb{F}$ belongs to the multi-ratio 
	$\left[ \begin{array}{ccc} P & Q & S\\ U &  R & T \end{array} \right] \subset \FF$ 
	if there exist elements $\boldsymbol{p} ,\boldsymbol{q},\boldsymbol{s}  \in\VV$ such that
	\begin{gather*}
	P = \langle \boldsymbol{p} \rangle , \qquad Q = \langle \boldsymbol{q} \rangle ,
	\qquad S = \langle \boldsymbol{s} \rangle ,
	 \qquad R = \langle \boldsymbol{p} + \boldsymbol{q}  \rangle , \qquad 
	 \qquad T = \langle \boldsymbol{p} + \boldsymbol{s}  \rangle , \qquad \\
	 U = \langle \boldsymbol{p} + \boldsymbol{q} a \rangle , \qquad
	 U = \langle \boldsymbol{p} + \boldsymbol{s} b \rangle \qquad m = ab^{-1}.  
	\end{gather*}
\end{Def}
\begin{Prop} \label{prop:m-alg}
	If $\boldsymbol{k} ,\boldsymbol{l}$ are two independent elements of $ \in\VV$, and 
	$p,q,r,s,t,u$ are six distinct elements of $\FF$ then
	\begin{equation} \label{eq:m-alg}
	\left[ \begin{array}{ccc} 
	\langle \boldsymbol{k} + \boldsymbol{l}p \rangle & 
	\langle \boldsymbol{k} + \boldsymbol{l} q \rangle & 
	\langle \boldsymbol{k} + \boldsymbol{l} s \rangle \\ 
	\langle \boldsymbol{k} + \boldsymbol{l}u \rangle &   
	\langle \boldsymbol{k} + \boldsymbol{l} r \rangle &
	\langle \boldsymbol{k} + \boldsymbol{l} t \rangle \end{array} \right] =
	\left[ (r-p)^{-1} (q-r) (q-u)^{-1} (s-u) (s-t)^{-1} (t-p)\right].
	\end{equation}
\end{Prop}
\begin{proof}
	Let the vectors $\boldsymbol{p} ,\boldsymbol{q}  \in\VV$ be such as in Def.~\ref{def:multi-ratio}, define the factors $\lambda,\mu,\sigma, \rho \in \FF\setminus\{0\} $ for points $P,Q,R,U$
	by
	\begin{equation*}
	(\boldsymbol{k} + \boldsymbol{l}p)\lambda = \boldsymbol{p}, \qquad
	(\boldsymbol{k} + \boldsymbol{l}q)\mu = \boldsymbol{q}, \qquad
	(\boldsymbol{k} + \boldsymbol{l}r)\sigma = \boldsymbol{p} + \boldsymbol{q} , \qquad
	(\boldsymbol{k} + \boldsymbol{l}u)\rho = \boldsymbol{p} + \boldsymbol{q} a .
	\end{equation*} 
Elimination of the factors $\mu,\sigma, \rho$ leads to the relation 
\begin{equation*}
a = \lambda^{-1} (r-p)^{-1} (q-r) (q-u)^{-1} (u-p) \lambda.
\end{equation*}
Similar reasoning, but for points $P,S,T,U$ leads to similar relation 
\begin{equation*}
b = \lambda^{-1} (t-p)^{-1} (s-t) (s-u)^{-1} (u-p) \lambda,
\end{equation*}
which, combined with the previous one, concludes the proof.
\end{proof}
\begin{Rem}
When the division ring $\FF$ is commutative, our definition of the multi-ratio 	$\left[ \begin{array}{ccc} P & Q & S\\ U &  R & T \end{array} \right] \in \FF$ reduces to  the product 
	$\left[ \begin{array}{cc} P & Q \\ U &  R  \end{array} \right] \left[ \begin{array}{cc} P & S\\ U &   T \end{array} \right]^{-1}$ of two cross-ratios.
\end{Rem}
\begin{Prop} \label{prop:mult-inv}
	The multi-ratio is an invariant of the collineations induced by linear transformations of the space $(\FF,\VV)$.
\end{Prop}
\begin{proof}
	Fix coordinate system, like in Proposition~\ref{prop:m-alg}, and use the fact, that such collineations are generated by affine transformations $\phi\mapsto a \phi  + b$ and inversions $\phi\mapsto \phi^{-1}$
	\begin{equation*}
	(a \phi + b)(c\phi + d)^{-1} = ac^{-1} + (b-ac^{-1}d)(c\phi +d)^{-1}.
	\end{equation*}
	By direct calculation, both transformations preserve the multi-ratio \eqref{eq:m-alg} understood as the class of conjugate elements.
\end{proof}

\subsection{Quadrangular set of points on projective line}
A \emph{complete quadrangle} is a projective figure formed by four points (vertices) $A,B,C,D$ in the
plane, no three of which are collinear, and the six distinct lines (sides) that are
produced by joining them pairwise. The intersection points $AB,\dots , CD$ of the six lines with a line 
not incident with vertices of the quadrangle form the \emph{quadrangular set} \cite{Staudt}, see Figure~\ref{fig:quadrangular-set}. 
\begin{figure}
	\begin{center}
		\includegraphics[width=12cm]{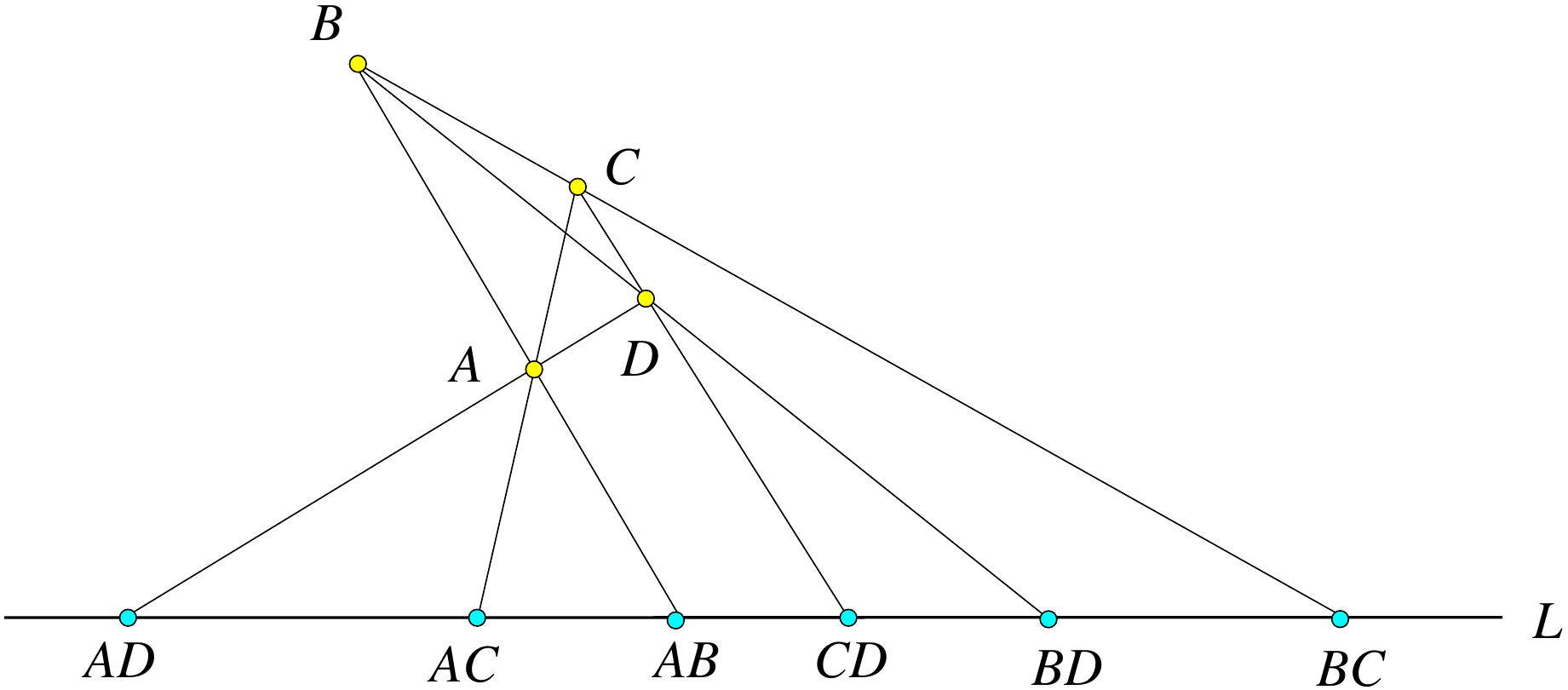}
	\end{center}
	\caption{Quadrangular set of points} 
	\label{fig:quadrangular-set}
\end{figure}
It is known that, by Desargues theorem, any five points of the quadrangular set (labeling fixed) determine uniquely the sixth point of the set. Moreover, collineations map quadrangular sets into quadrangular sets.

\begin{Rem}
	The ordering of the points is important, up to permutation of the letters $A,B,C,D$. By combinatorial arguments one can show that given five generic points of the projective line there are $30$ positions of the sixth point such that for appropriate ordering the six points form a quadrangular set. 
\end{Rem}
\begin{Prop} \label{prop:Veblen-alg}
	The six distinct points $AB,\dots ,CD$ form a quadrangular set if and only if their their non-homogeneous cooridinates $\Phi_{AB}, \dots , \Phi_{CD}$ satisfy the multi-ratio condition
	\begin{equation} \label{eq:quad-mr}
	(\Phi_{AD} - \Phi_{AC}) (\Phi_{AD} - \Phi_{AB})^{-1}  
	(\Phi_{BD} - \Phi_{AB}) (\Phi_{BD} - \Phi_{BC})^{-1}  
	(\Phi_{CD} - \Phi_{BC}) (\Phi_{CD} - \Phi_{AC})^{-1}   = 1 .
	\end{equation}
\end{Prop}
\begin{proof}
 Given five points of the set, we reconstruct the planar quadrilateral which allows to obtain the sixth point. It is known~\cite{VeblenYoung} that the construction is independent on the freedom in choice of the quadrilateral. 
Because we were not able to find the multi-ratio characterization of the quadrangular sets in the  non-commutative case we present its detailed derivation.  

	In the general case fix coordinate system on the line  $L = \{ (a,0) | a\in \FF\}$. Choose an arbitrary point $A\neq L$, which can be given then non-homogeneous coordinates $(0,1)$. The last freedom in the construction is the choice of a point $B$ on the line $\langle A, AB \rangle$, which fixes its coordinates 
	\begin{equation*}
	(\Phi_B^1, \Phi_B^2) = (\Phi_{AB} (1-\sigma),\sigma) , \qquad \sigma\in\FF \setminus \{0,1 \} .
	\end{equation*}
	Then the coordinates of the point $C = \langle A,AC \rangle \cap  \langle B ,BC \rangle$ are
	\begin{equation*}
	(\Phi_C^1, \Phi_C^2) = (\Phi_{AC} (1-\mu),\mu) , \qquad \mu =
	\left[ \Phi_{AC} - \Phi_{AB} + (\Phi_{AB} - \Phi_{BC}) \sigma^{-1} \right]^{-1} (\Phi_{AC} - \Phi_{BC}) ;
	\end{equation*}
	notice identity
	\begin{equation} \label{eq:mu-sigma}
	\mu^{-1} - 1 = (\Phi_{AC} - \Phi_{BC})^{-1} (\Phi_{AB} - \Phi_{BC}) (\sigma^{-1} - 1).
	\end{equation}
	Similarly, the coordinates of the point $D = \langle A, AD \rangle \cap  \langle C, CD \rangle$ read
	\begin{equation*}
	(\Phi_D^1, \Phi_D^2) = (\Phi_{AD} (1-\lambda),\lambda) , \qquad \lambda =
	\left[ \Phi_{AD} - \Phi_{AC} + (\Phi_{AC} - \Phi_{CD}) \mu^{-1} \right]^{-1} (\Phi_{AD} - \Phi_{CD}) ,
	\end{equation*}
	and then
	\begin{equation} \label{eq:lambda-mu}
	\lambda^{-1} - 1 = (\Phi_{AD} - \Phi_{CD})^{-1} (\Phi_{AC} - \Phi_{CD}) (\mu^{-1} - 1),
	\end{equation}
	moreover
	\begin{equation} \label{eq:mu-lambda}
	(\Phi_{CD} - \Phi_{AC})^{-1} (\Phi_{AD} - \Phi_{AC}) = \mu^{-1}(\mu - \lambda)(1-\lambda)^{-1}.
	\end{equation}
	Finally, the non-homogeneous coordinates $(\Phi_{BD},0)$ of the point $BD = L \cap \langle B,D \rangle $ are given by
	\begin{equation} \label{eq:BD}
	\Phi_{BD} = \Phi_{AB} (1-\sigma) + \left[ \Phi_{AD}(1-\lambda) - \Phi_{AB}(1-\sigma) \right] (\sigma - \lambda)^{-1}\sigma,
	\end{equation}
	what implies
	\begin{equation} \label{eq:lambda-sigma}
 (\Phi_{BD} - \Phi_{AB})^{-1} (\Phi_{AD} - \Phi_{AB}) = \sigma^{-1}(\sigma - \lambda)(1-\lambda)^{-1}.
	\end{equation}
	Equations \eqref{eq:mu-sigma}-\eqref{eq:lambda-sigma} give then
	\begin{gather*}
(\Phi_{BD} - \Phi_{BC})   -
 (\Phi_{CD} - \Phi_{BC}) (\Phi_{CD} - \Phi_{AC})^{-1} (\Phi_{AD} - \Phi_{AC})         (\Phi_{AD} - \Phi_{AB}) ^{-1} (\Phi_{BD} - \Phi_{AB}) = \\
 \left[  (\Phi_{AD} - \Phi_{CD}) (\lambda^{-1}-1) - 
 (\Phi_{AB} - \Phi_{BC}) (\sigma^{-1}-1) + 
 (\Phi_{CD} - \Phi_{BC}) (\mu^{-1}-1) \right] \lambda(\sigma-\lambda)^{-1}\sigma = 0, 
	\end{gather*}
	which concludes the first part of the proof.

Because equation~\eqref{eq:quad-mr} is uniquely solvable for any of its six points, once other five are given, and by the analogous property of the quadrangular set, the condition described by the equation completely characterizes quadrangular sets of the projective line.		
\end{proof}

\section{Desargues maps into projective line} \label{sec:Des-line}
\subsection{The Veblen configuration and the multi-ratio}
Consider Veblen (or Menelaus~\cite{KoSchief-Men}) configuration $(6_2,4_3)$ in projective space, i.e. six points and four lines with two lines incident with each point, and three points incident with each line, see Figure~\ref{fig:Veblen}. We label points of the configuration by two element subsets of the four element set $\{A,B,C,D\}$, and lines by single elements of the same set. A point is incident with the line if its label contains the label of the line.
\begin{figure}
	\begin{center}
		\includegraphics[width=7cm]{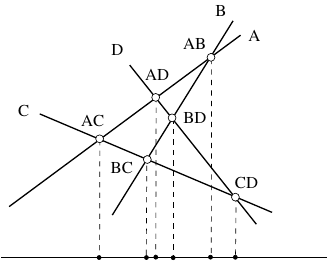}
	\end{center}
	\caption{Veblen configuration} 
	\label{fig:Veblen}
\end{figure}

Let us present an algebraic description of the Veblen configuration, which can be considered as a non-commutative version of the theorem of  Menelaus~\cite{Coxeter-Greitzer}. 
\begin{Prop}
	Given points $AB\in\langle BC,BD \rangle$, $AD\in\langle BD,CD \rangle$, $AC\in\langle BC,CD \rangle$ on three sides of the triangle $BC,BD,CD$, and distinct from its vertices. These three points are collinear if and only if 
	the corresponding proportionality coefficients $a,b,c\in\FF\setminus\{ 0,1 \}$ between their non-homogeneous coordinates, as defined by
	\begin{gather} \label{eq:abc-lin}
	(\bPhi_{BC} - \bPhi_{AB}) = (\bPhi_{BD} - \bPhi_{AB}) a, \\
	(\bPhi_{CD} - \bPhi_{AD}) = (\bPhi_{BD} - \bPhi_{AD}) b, \qquad
	(\bPhi_{BC} - \bPhi_{AC}) = (\bPhi_{CD} - \bPhi_{AC}) c, \nonumber
	\end{gather}
	satisfy condition
	\begin{equation}
	a=bc. \label{eq:a=bc}
	\end{equation}
\end{Prop}
\begin{proof} 
	To show that the collinearity implies condition \eqref{eq:a=bc} assume that the vectors $\bPhi_{AB}, \bPhi_{AC}, \bPhi_{AD}$, as calculated from the above linear relations, satisfy the constraint of the form
	\begin{equation*}
	\bPhi_{AB} - \bPhi_{AC} = (\bPhi_{AD} - \bPhi_{AC} )\lambda .
	\end{equation*}
	The linear independence of vectors $ \bPhi_{BD} - \bPhi_{CD} $ and  $ \bPhi_{BC} - \bPhi_{CD} $ implies then
	\begin{gather*}
	\lambda = 1 - (c-1) (a-1)^{-1} = (b-1)c(a-1)^{-1},
	\end{gather*}
	which gives equation \eqref{eq:a=bc}.
	
	From the other side, insert the condition \eqref{eq:a=bc} into first of the above linear equations, which in conjunction with other two gives 
	\begin{equation*}
	(\bPhi_{AB} - \bPhi_{AD} ) (1  - bc)= (\bPhi_{AC} - \bPhi_{AD} )(1-c),
	\end{equation*}
	thus showing the collinearity.
\end{proof}
\begin{Cor} \label{cor:Veblen-multiratio}
	Assume that for fixed coordinate number $i$ all components $\Phi_{AB}^i,\dots , \Phi_{CD}^i$ of the points of the Veblen configuration are distinct (see Figure~\ref{fig:Veblen}), then 
	the components satisfy the following multi-ratio condition
	\begin{equation*} \label{eq:multiratio}
	(\Phi_{CD}^i - \Phi_{AC}^i) (\Phi_{CD}^i - \Phi_{AD}^i)^{-1} 
	(\Phi_{BD}^i-\Phi_{AD}^i) (\Phi_{BD}^i - \Phi_{AB}^i)^{-1} (\Phi_{BC}^i-\Phi_{AB}^i) 
	(\Phi_{BC}^i-\Phi_{AC}^i)^{-1} = 1 .
	\end{equation*}
\end{Cor}
\begin{proof}
	Insert expressions
	\begin{align*}
	a = & (\Phi_{BD}^i - \Phi_{AB}^i)^{-1} (\Phi_{BC}^i-\Phi_{AB}^i) , \\
	b = & (\Phi_{BD}^i - \Phi_{AD}^i)^{-1} (\Phi_{CD}^i-\Phi_{AD}^i) , \\
	c = & (\Phi_{CD}^i - \Phi_{AC}^i)^{-1} (\Phi_{BC}^i-\Phi_{AC}^i) , 
	\end{align*}
	into the condition \eqref{eq:a=bc}.
\end{proof}
We conclude this Section with a result, which justifies the statement that \emph{quadrangular sets should be considered as Veblen configurations in the geometry of projective line}.
\begin{Prop} \label{prop:Veblen-proj-qs}
	In the plane of the Veblen configuration consider point $O$ not on lines of the configuration. The intersection points of lines joining $O$ to vertices of the configuration with an arbitrary line not incident with $O$ form a quadrangular set.
\end{Prop}
\begin{proof}
	Take the point $O$ as the first vertex of the quadrangle, fix a line of the Veblen configuration, and use three remaining points of the configuration as three remaining vertices of the quadrangle. On the line we have built then a quadrangular set. The lines joining the points of the Veblen configuration with point $O$ are the lines joining $O$ to points of the quadrangular set. Any transversal section of the lines by another line gives six points perspective with the quadrangular set. Because such transformations map quadrangular sets into quadrangular sets~\cite{VeblenYoung} we obtain the statement.
\end{proof}
\begin{Rem}
	The case when $O$ is a point at infinity and the line is a coordinate line is actually visualized in Figure~\ref{fig:Veblen}.
\end{Rem}

\subsection{Desargues maps}

From point of view of difference equations usually one considers maps of $\ZZ^N$ lattice. Recently integrable systems on other regular lattices are also of some interest, see for example \cite{DNS-hex}. In particular, the Desargues maps, although initially defined on multidimensional integer lattice, allow for an interpretation~\cite{Dol-AN} as maps from multidimensional root lattice of type $A$. Such an approach from the very beginning takes into account the corresponding affine Weyl group symmetry of the discrete KP system. 

Recall that the $N\geq 2$-dimensional root lattice $Q(A_N)$ 
is generated by vectors along the edges of regular $N$-simplex in Euclidean space~\cite{ConwaySloane,MoodyPatera}. Vertices of the lattice  tessellate the space into $N$ types of convex polytopes $P(k,N)$, $k=1,\dots,N$, called \emph{ambo-simplices}. It is known that the corresponding  affine Weyl group $W(A_N)$ acts on the Delaunay tiling by permuting tiles within each class $P(k,N)$. The $1$-skeleton of $P(k,N)$ is 
the so called Johnson graph $J(N+1,k)$: its vertices are labeled by 
$k$-point subsets of
$\{1,2, \dots , N+1\}$, and edges are the pairs of such sets with
$(k-1)$-point intersection. 

The tiles $P(1,N)$ are congruent to the initial $N$-simplex which generates the vertices of the root lattice. We color its faces $P(1,2)$ in black. The faces $P(2,2)$ of the simplex $P(N,N)$ we color white. Then $P(2,3)$ which is regular octahedron has four black and four white triangular faces, see Figure~\ref{fig:Des-map}.
\begin{Def}
	By Desargues map $\phi \colon Q(A_N) \to \PP(\FF,\VV)$ we mean a map for which images of vertices of simplices $P(1,2)$ with black triangular faces are collinear.
\end{Def}
\begin{figure}
	\begin{center}
		\includegraphics[width=8cm]{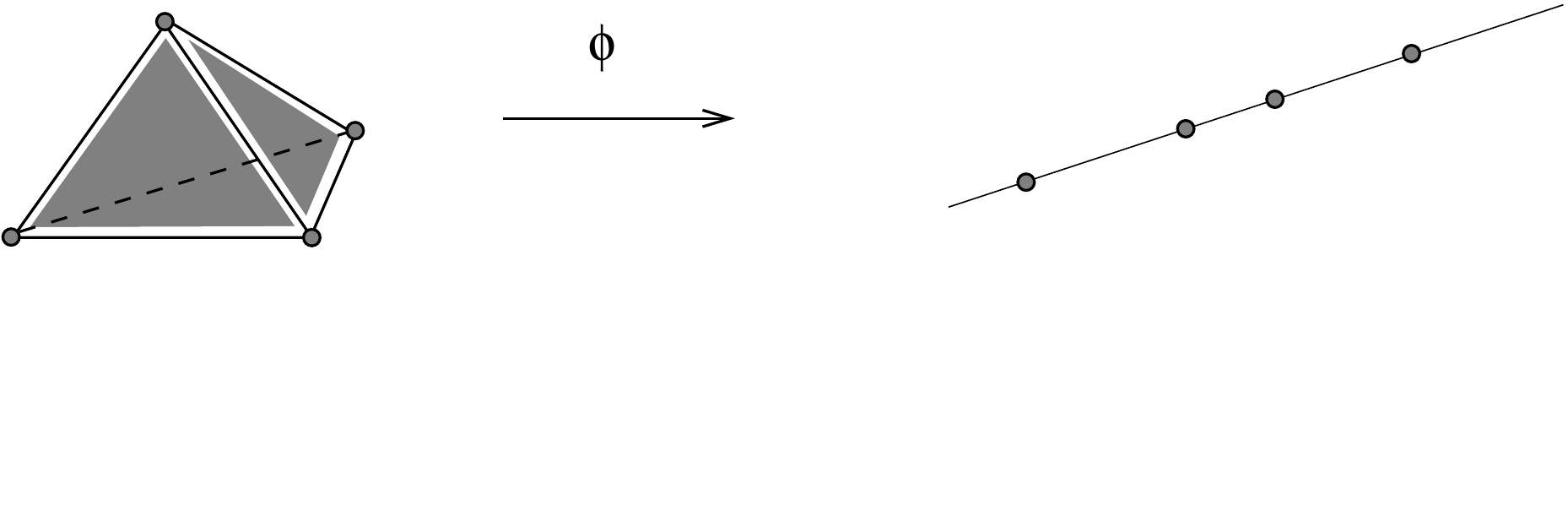}
		
		\includegraphics[width=9cm]{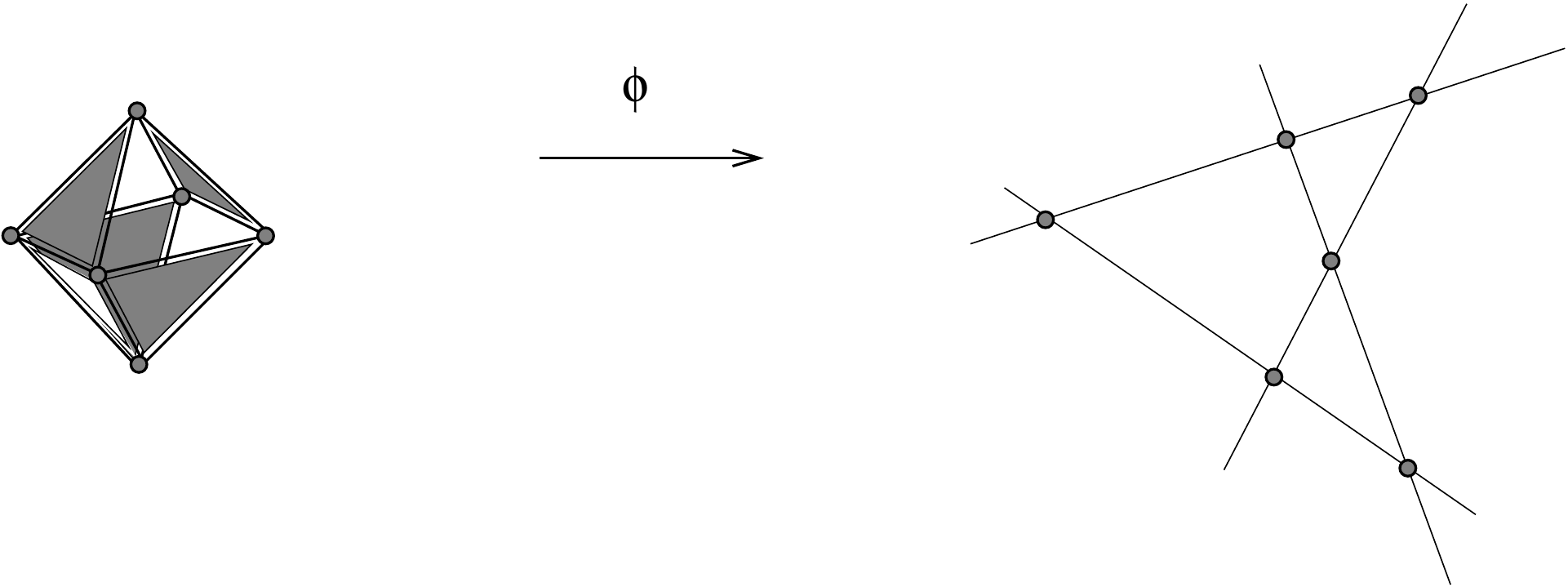}
	\end{center}
	\caption{Desargues map} 
	\label{fig:Des-map}
\end{figure}
To avoid degenerations it is implicitly assumed that images of vertices of simplices $P(2,2)$ with white triangular faces are in generic position. 
Then the octahedra $P(2,3)$ are mapped into Veblen configurations, see Figure~\ref{fig:Des-map}. Moreover, after introduction of appropriate $\ZZ^N$ coordinates in the lattice the non-homogeneous coordinates of the map satisfy equation~\eqref{eq:NCSKP}, compare with Corollary~\ref{cor:Veblen-multiratio}. Guided by Proposition~\ref{prop:Veblen-proj-qs} we give our answer to Question 2.
\begin{Def}
	By Desargues map of root lattice $ Q(A_N)$  into a projective line we mean a map for which images of vertices of ambo-simplices $P(2,3)$ are quadrangular sets with labeling induced by that of Johnson graph $J(4,2)$.
\end{Def}

\section{Desargues maps in Moebius chain geometry}
\label{sec:CG}
Below we describe the concept of Moebius chain geometry of a projective line, which allows for a construction of quadrangular sets which satisfy certain additional property, adjusted to the structure. This will give our answer to Question~1.
\subsection{The concept of Moebius chain geometry}
Assume that the division ring $\FF$ contains a proper subfield $\KK$ in its center. The division ring $\FF$ can be considered then as a division $\KK$-algebra. Correspondingly, the projective line  over $\FF$ inherits additional structure, best described within the concept of chain geometry \cite{Benz,Herzer,BlunckHerzer}. 
Define the chains as images of the canonically embedded $\KK$-line (called the standard chain) under action of the group of collineations induced by linear maps of the $\FF$-line. Points are called cocatenal if they belong to a common chain.
\begin{Rem}
The simplest example for $\KK=\RR\subset\CC=\FF$ is the classical conformal Moebius geometry of circles (as chains) in the Riemann sphere (complex line or conformal plane). The notion of chain geometry applies actually to any $\KK$-algebra, however we will be dealing exclusively with division algebras (i.e. with the so called~\cite{Herzer} Moebius chain geometries). 
\end{Rem}
The $\KK$-vector space $\FF$ can be given natural affine space structure. The straight lines are the chains which pass through the infinity point of the projective line $ \FF\cup \{ \infty \}$. Notice that the notion of "straight line" depends actually on the particular choice of the infinity point. Two chains ar called tangent in point $P$ when after sending $P$ to infinity the chains are parallel in corresponding affine space.
Notice following results of the Moebius chain geometry:
\begin{enumerate}
\item Any three distinct points of the $\FF$-line are contained in exactly one chain.
\item Four distinct points are cocatenal if and only if their cross-ratio is well defined element of $\KK \setminus \{ 0,1 \}$.
\item (Miquel condition) Given four chains $\mathcal{C}_i$, $i=1,2,3,4$, no three of which have a common point, but 
$\mathcal{C}_i \cap \mathcal{C}_{i+1} = \{ P_i, Q_i \}$ for every $i$ (subscripts are taken modulo $4$). Then the four points $P_i$ are cocatenal if and only if the four points $Q_i$ are cocatenal.
\end{enumerate}
\begin{Rem}
	One can consider~\cite{BlunckHavlicek} generalizations of the Moebius chain geometry for which $\KK$ is a subdivision ring of $\FF$ not necessarily contained in its center (like for example in the case $\CC\subset \HH$, where chains are two dimensional spheres $S^2$ in $S^4 \equiv \HH\cup\{\infty\}$). However such a generalization may violate some of the properties above.
\end{Rem}

\subsection{Moebius--Veblen configuration}
Let us present an analogue of the Veblen axiom/theorem, see also \cite{BlunckHerzer} for its general version (in a slightly different formulation) for chain geometries.
\begin{figure}
	\begin{center}
		\includegraphics[width=16cm]{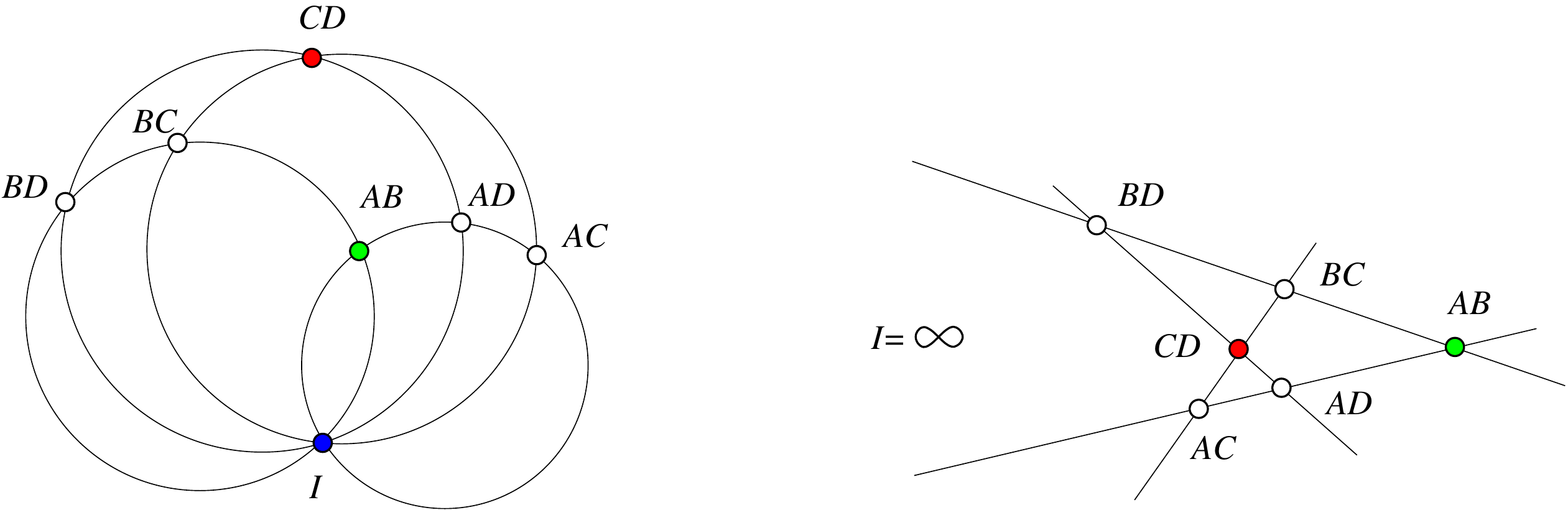}
	\end{center}
	\caption{Moebius--Veblen configuration in Moebius chain geometry} 
	\label{fig:Chain-Veblen}
\end{figure}
\begin{Prop} \label{prop:c-Veblen}
Given five distinct points $AB, AC, AD, BC, BD$ of the Moebius chain space such that the chains $\mathcal{C}(AB, AC, AD)$ and $\mathcal{C}(AB, BC, BD)$ have in common additional point $I\neq AB$. Then also the chains $\mathcal{C}(I, AC, BC)$ and 
$\mathcal{C}(I, AD, BD)$ have in common additional point $CD\neq I$, or alternatively the chains are tangent in $I$. 
\end{Prop}
\begin{proof}
By sending the intersection point $I$ to infinity, see Figure~\ref{fig:Chain-Veblen},  we obtain five points of the classical Veblen configuration in the affine 
space $(\KK,\FF)$ what allows to construct the sixth point of the "straightened configuration", and then eventually to go back to the original one. 
\end{proof}
\begin{Rem}
	Notice that in the classical Moebius geometry of the complex projective line the assumption about existence of the point $I$ is not needed. However this assumption is essential, if we would like to perform the construction in the case of quaternionic projective line with $\RR\subset \HH$. 
\end{Rem}
The points  $AB,AC, AD, CD,  BC, BD$ will be called ordinary points of the Moebius--Veblen configuration, while $I$ is called the infinity point of the configuration. 
\begin{Prop}
	In the Moebius--Veblen configuration the constructed point $CD$ is the sixth quadrangular point of the initial five ordinary points in the standard labeling.
\end{Prop}
\begin{proof}
In the "straightened configuration" apply Proposition~\ref{prop:Veblen-alg} to get  coefficients $a,b,c$ which belong to $\KK\setminus \{0,1 \}$ and satisfy condition \eqref{eq:a=bc}. But now one can eliminate the coefficients directly on the level of 
equations~\eqref{eq:abc-lin} in order to get the multi-ratio condition \eqref{eq:quad-mr} in $\FF$. Proposition~\ref{prop:mult-inv}, which gives invariance of the condition with respect to projective collineations of the $\FF$-line, and Proposition~\ref{prop:Veblen-alg} imply the statement.
\end{proof}
\begin{Rem}
	Even if we do not have the point $I$ to our disposal the sixth point $CD$ of the quadrangular set exists, by the general construction described in Section~\ref{sec:line}. Then one can consider also the corresponding initial chains $\mathcal{C}(AB, AC, AD)$, $\mathcal{C}(AB, BC, BD)$, which contain point $AB$, and the resulting chains $\mathcal{C}(AC, BC, CD)$ and $\mathcal{C}(AD, BD, CD)$ containing the point $CD$. These new chains can be given only \emph{after}  construction of the point. 
\end{Rem}
Finally we give 
our answer to Question~1 presenting the special type of quadrangular sets for which the construction of the sixth point follows from Moebius chain geometry principles, thus incorporating the conformal geometry interpretation of the discrete Schwarzian Kadomtsev--Petviashvili equation into the Desargues map approach.
\begin{Def}
	By Moebius quadrangular set we mean the six ordinary points of the Moebius--Veblen configuration with the labeling as in Proposition~\ref{prop:c-Veblen}. 
\end{Def}
It is well known~\cite{Clifford,KoSchief-Men} that the  Moebius--Veblen configuration can be supplemented by four chains $\mathcal{C}(AB, BC, AC)$, $\mathcal{C}(AB, BD, AD)$, $\mathcal{C}(AC, CD, AD)$ and $\mathcal{C}(BC, CD, BD)$ which then all intersect in the so called Clifford point. The new circle-point configuration of $8$ points and $8$ circles, with each point/circle incident with $4$ circles/points, is called the Clifford configuration. Actually, as it was described in~\cite{Doliwa-isoth}, this result in an  equivalent version was known already to Miquel~\cite{Miquel}.

\section{Conclusion}
The projective structure of the line and the notion of quadrangular sets can be used to provide geometric meaning to non-commutative discrete Schwarzian KP system. When the thiner structure of the underlying division ring is considered then the theory becomes more intriguing. Up to now the case of complex and quaternionic Moebius spaces (with the subfield of real numbers) have been throughly examined~\cite{KoSchief-Men,KoSchief-qSKP}. There are known works, see for example~\cite{BaMaSe}, where the Miquel condition has been used to study particular algebras of quantum integrable systems.
 We expect that the chain geometry may become useful platform to investigate other quantum algebras of mathematical physics.

\section*{Acknowledgments}
Both authors acknowledge numerous discussions with Andrzej Matra\'{s} on various aspects of chain geometries.
The research of A.D. was supported by National Science Centre, Poland, under grant 2015/19/B/ST2/03575 \emph{Discrete integrable systems -- theory and applications}.

\bibliographystyle{amsplain}

\providecommand{\bysame}{\leavevmode\hbox to3em{\hrulefill}\thinspace}

\end{document}